\documentclass{article}

\usepackage{paralist}
\usepackage{amsmath}
\usepackage{amsthm}
\usepackage{amssymb}
\usepackage{amsfonts}
\usepackage{array}
\usepackage{paralist}
\usepackage{fullpage}

\newtheorem{theorem}{Theorem}[section]

\newtheorem{lemma}[theorem]{Lemma}

\newtheorem{definition}[theorem]{Definition}

\newcommand{\RR}{\mathbb R}

\newcommand{\partdiff}[2]{\frac{\partial {#1}}{\partial {#2}}}
\newcommand{\secdiff}[2]{\frac{\partial^2 {#1}}{\partial {#2}^2}}
\newcommand{\mixdiff}[3]{\frac{\partial^2 {#1}}{{\partial {#2}}{\partial {#3}}}}

\newcommand{\E}{\mbox{\bf E}}

\newcommand{\cM}{{\cal M}}

\newcommand{\cF}{{\cal F}}

\newcommand{\cR}{{\cal R}}

\def\b1{{\bf 1}}

\def\bg{{\bf g}}

\def\bg{{\bf g}}

\def\bx{{\bf x}}
\def\by{{\bf y}}

\title{An approximately truthful-in-expectation mechanism for
  combinatorial auctions using value queries} 

\author{Shaddin Dughmi\thanks{Microsoft Research, Redmond, WA} \and Tim Roughgarden\thanks{Stanford University, Stanford, CA} \and Jan Vondr\'ak\thanks{IBM Almaden Research Center, San Jose, CA} \and Qiqi Yan\thanks{Stanford University, Stanford, CA}}

\begin{document}

\maketitle

\section{Introduction}

This manuscript presents an alternative implementation of the
truthful-in-expectation (TIE) mechanism of Dughmi, Roughgarden and Yan
\cite{DRY11} for combinatorial auctions. Recall that in a {\em
  combinatorial auction}, $m$ goods get allocated to $n$
bidders.  Each bidder~$i$ has a private valuation~$v_i$ that describes
its value~$v_i(S)$ for each subset~$S$ of goods.  The social welfare
of an allocation is the sum of the bidders' values for the goods
received.

The mechanism of \cite{DRY11}
is presented in a ``lottery-value" oracle model, where each bidder
can be queried about his valuation by means of the following query:
given a vector of probabilities over items $\bx \in [0,1]^m$, what is
the expected value $\E[v_i(\hat{\bx})]$, where $\hat{\bx}$ is obtained
by independently rounding each coordinate of $\bx$ to $0$ or $1$ with
probability $x_i$. Such queries can be answered efficiently for
certain valuation functions (in particular {\em coverage} functions),
and this oracle model is a convenient framework for the presentation
of the mechanism of \cite{DRY11}. On the other hand, lottery-value
queries are \#P-hard to answer for the class of matroid rank functions
(see Section~\ref{sec:SharpP}), and hence one can ask how realistic
this model is in general. The purpose of this manuscript is to show that
the model is ``approximately realistic'' in the sense that the mechanism
of \cite{DRY11} can be implemented in the (weaker) value oracle model
at the cost of relaxing the concept of truthfulness in expectation to
{\em approximate truthfulness in expectation}. (Here, we mean
appoximation within an arbitrarily small error, in the sense of an
FPTAS.) 
In particular, we show that the maximal in distributional range (MIDR)
allocation rule of \cite{DRY11} can be implemented as an approximately MIDR
allocation rule in the value oracle model, and then we present a
blackbox transformation from approximately MIDR allocation rules to
approximately TIE mechanisms.

First, let us define the approximate variants of MIDR and truthfulness
in expectation. The exact variants are obtained by setting
$\epsilon=0$. 

\begin{definition}
An allocation rule is {\em $(1-\epsilon)$-approximately
  maximal-in-distributional range} (or {\em $(1-\epsilon)$-MIDR}) if
  there is a range of distributions 
over outcomes $\cR$ such that, for every input,
the mechanism returns an outcome that is sampled from a distribution
$D^* \in \cR$ that $(1-\epsilon)$-approximately maximizes the expected
social welfare 
$\E_{\omega \sim D}[\sum_i v_i(\omega)]$ over all distributions $D \in
\cR$. 
\end{definition}

We define approximate truthfulness in expectation in terms of {\em
  relative utility error}. This means that truth-telling costs a
bidder at most an $\epsilon$ fraction of his maximum-possible utility. 

\begin{definition}
A mechanism with allocation and payment rules $A$ and $p$ is
{\em $(1-\epsilon)$-approximately truthful-in-expectation} (or
{\em $(1-\epsilon)$-TIE}) if, for every bidder~$i$, (true) valuation
function~$v_i$, (reported) valuation function~$v'_i$, and (reported)
valuation functions~$v_{-i}$ of the other bidders,
\begin{equation}\label{eq:approx-truthful}
\E[v_i(A(v_i,v_{-i})) - p_i(v_i,v_{-i})] \geq (1-\epsilon) \E[
  v_i(A(v'_i,v_{-i})) - p_i(v'_i,v_{-i})]. 
\end{equation}
The expectation in~\eqref{eq:approx-truthful} is over the coin flips
of the mechanism.  
\end{definition}

The class of valuations of interest here is the following (as in \cite{DRY11}).

\begin{definition}
\label{def:WMRS}
A function $v:2^{[m]} \rightarrow \RR_+$ is a {\em weighted matroid rank
sum} if there are matroids $\cM_1,\ldots,\cM_k$ 
and weights $\alpha_1,\ldots,\alpha_k \geq 0$ such that
$$ v(S) = \sum_{i=1}^{k} \alpha_i r_{\cM_i}(S), $$
where $r_{\cM_i}$ is the rank function of matroid $\cM_i$.
\end{definition}
This definition also captures positive combinations of {\em weighted
rank functions}, as the cones generated by 
weighted and unweighted rank functions of matroids coincide. 

The allocation rule of \cite{DRY11} is maximal-in-distributional-range,
provides a $(1-1/e)$-approximation to the social welfare, 
and the corresponding TIE mechanism can be implemented in expected
polynomial time provided the bidders' valuations are weighted
matroid rank sums and support lottery-value queries. Our goal here is
to prove the following. 

\begin{theorem}
\label{thm:approx-TIE}
For every $\epsilon = {1}/{poly(m,n)}$, there is a 
$(1-\epsilon)$-TIE
mechanism that achieves a $(1-1/e-\epsilon)$-approximation to the
social welfare in
combinatorial auctions with weighted-matroid-rank-sum valuations 
that runs in polynomial time in the value oracle model. 
\end{theorem}

We prove this claim in two steps. First, we prove the following.

\begin{theorem}
\label{thm:approx-MIDR}
For every $\epsilon = {1}/{poly(m,n)}$, there is a 
$(1-\epsilon)$-MIDR allocation rule that achieves a
$(1-1/e-\epsilon)$-approximation to the social welfare in
combinatorial auctions with weighted-matroid-rank-sum bidders and runs
in polynomial time in the value oracle model. 
\end{theorem}

Then, we prove the following general reduction.

\begin{theorem}
\label{thm:MIDR->TIE}
For every $(1-\epsilon)$-MIDR allocation rule that is a
$c$-approximation to the 
social welfare in combinatorial auctions with bidders' valuations
restricted to a set $\cal C$, there is a $(1-\epsilon')$-TIE mechanism
that is a $(c-1/poly(m,n))$-approximation to the social welfare in 
combinatorial auctions with bidders' valuations in~$\cal C$, where
$\epsilon'=\epsilon \cdot poly(m,n)$.
\end{theorem}

Therefore, by selecting a sufficiently (polynomially) small $\epsilon$
in Theorem~\ref{thm:approx-MIDR}, we can achieve an arbitrarily
(polynomially) small error $\epsilon$ in
Theorem~\ref{thm:approx-TIE}. 

\section{An Approximately MIDR Allocation Rule}

In this section, we prove Theorem~\ref{thm:approx-MIDR}.
We use a variant of the mechanism of \cite{DRY11}. Instead of
sophisticated convex optimization techniques, which seem necessary to
find the exact optimum over the range, we use a simple local search
that guarantees that we get arbitrarily close to the optimum. 
We begin with some definitions. 

\begin{definition}
For a combinatorial auction with $m$ items and $n$ bidders with
valuations $v_i:2^{[m]} \rightarrow \RR_+$, the {\em aggregate
valuation function} $f:2^{[n] \times [m]}  \rightarrow \RR_+$ is 
$$ f(S) = \sum_{i=1}^{n} v_i(\pi_i(S)), $$
where $\pi_i(S) = \{ j: (i,j) \in S\}$. We define $F:[0,1]^{[n]
  \times [m]}  \rightarrow \RR_+$ to be the multilinear extension of
$f$ (see also~\cite{V08}),
and $P$ to be the polytope of fractional allocations:
$$ P = \left\{ \bx \in [0,1]^{[n] \times [m]}: \forall j: \, \sum_{i=1}^{n}
x_{ij} \leq 1 \right\}.$$
\end{definition}

The (integral) welfare maximization problem turns out to be equivalent
to $\max \{F(\bx): \bx \in P\}$.
This problem cannot be solved optimally, even for very
special classes of valuation functions.
In lieu of $F$, the authors of \cite{DRY11} use the
modified objective function 
$$ F^{exp}(x_{11}, x_{12}, \ldots, x_{nm}) = F(1-e^{-x_{11}}, 1- e^{-x_{12}},
\ldots, e^{-x_{nm}}).$$
Interestingly, the function $F^{exp}$ turns out to be {\em concave}
for a subclass of submodular functions, including weighted matroid
rank sums (see \cite{DRY11} for a proof).
This means that we can solve the problem $\max \{F^{exp}(\bx): \bx \in P\}$,
which means in effect optimizing over a certain
range of product distributions. Also, the optimum of this problem is
at least $(1-1/e)$ times the optimal social welfare.
Supplementing this MIDR allocation rule with suitable payments yields
a $(1-1/e)$-approximate, TIE mechanism~\cite{DRY11}. 

Here, we propose the following simple algorithm that solves the
problem $\max \{F^{exp}(\bx): \bx \in P\}$ 
near-optimally (in the sense of an FPTAS). 

\paragraph{Local Search Allocation Rule.}
\begin{itemize}
\item Initialize $\bx:=0$. Let $M$ be the maximum value of any singleton.
\item Let $\bg$ be an estimate of the gradient $\nabla F^{exp}(\bx)$, within additive error $\delta M$ in each coordinate,
where $\delta = \frac{\epsilon}{8 m^2 n^2}$.
As long as there is a point $\by \in P$ such that
$$ (\by - \bx) \cdot \bg > \frac12 \epsilon M, $$
update $\bx := \bx + \delta (\by - \bx)$.
\item Return an allocation randomly sampled from the distribution $\bx$.
\end{itemize} 

The required estimates of $\nabla F^{\exp}(\bx)$ can be
obtained in polynomial time by random sampling, with high
probability. (By {\em high   probability}, we mean $1 -
e^{-poly(m,n)}$ in this manuscript.  The coordinates of $\nabla
F^{exp}(\bx)$ are always in the interval $[0,M]$ --- by submodularity
--- and so this follows from standard Chernoff bounds
\cite{AlonSpencer}.)  Linear programming can be used to efficiently
find a suitable point~$\by$, or certify that no such point exists.

\subsection{The analysis}

We claim that this allocation rule runs in polynomial time and solves
the problem $\max \{F^{exp}(\bx): \bx \in P\}$ up to a $(1-\epsilon)$
factor with high probability, thus proving Theorem~\ref{thm:approx-MIDR}. 
In the following, we assume that the estimate
of $\nabla F^{exp}(\bx)$ obtained in each step is accurate within additive error $\delta M$,
which is possible to achieve with high probability over the run of the algorithm, via a polynomial number of samples.

We proceed via a series of claims.

\begin{lemma}
If the algorithm terminates, then with high probability
$$ F^{exp}(\bx) \geq (1-\epsilon) \max \{ F^{exp}(\bx): \bx \in P \}.$$
\end{lemma}

\begin{proof}
Let $\by$ be an optimal solution of $\max \{ F^{exp}(\bx): \bx \in P \}$.
When the algorithm terminates at $\bx$, we have $(\by-\bx) \cdot
\nabla F^{exp}(\bx) < \epsilon M$ (even accounting for the errors in
our estimate of $\nabla F^{exp}(\bx)$). By the concavity of $F^{exp}$,
$$ OPT - F^{exp}(\bx) = F^{exp}(\by) - F^{exp}(\bx) \leq (\by - \bx)
\cdot \nabla F^{exp}(\bx) \leq \epsilon M \leq \epsilon OPT.$$ 
\end{proof}

\begin{lemma}
In each iteration, with high probability, the value of $F^{exp}(\bx)$
increases by at least $\frac{\epsilon^2}{64 m^2 n^2}  M$. 
\end{lemma}

\begin{proof}
If the algorithm continues, we can assume that $(\by - \bx) \cdot
\nabla F^{exp}(\bx) > \frac14 \epsilon M$ 
(considering that the estimate of $\nabla F^{exp}(\bx)$ could be off
by $\delta M = \frac{\epsilon M}{8 m^2 n^2}$ in each coordinate). 
We also have bounds on how much the gradient can change when $\bx$
moves by a certain amount.
Specifically, for $(i,j) \neq (i',j')$,
$$ \Big| \mixdiff{F^{exp}}{x_{ij}}{x_{i'j'}}\Big| = \Big|
 e^{-x_{ij}-x_{i'j'}} \mixdiff{F}{x_{ij}}{x_{i'j'}} \Big| 
 \leq \Big|\mixdiff{F}{x_{ij}}{x_{i'j'}}\Big| \leq M
$$
and similarly
$$ \Big| \secdiff{F^{exp}}{x_{ij}} \Big| = \Big| e^{-x_{ij}}
\partdiff{F}{x_{ij}} \Big| \leq \Big| \partdiff{F}{x_{ij}} \Big| \leq
M,$$ 
by known properties of the multilinear extension \cite{V08,V09}. This
implies that for any $\bx'$ such that 
$|| \bx' - \bx ||_\infty \leq \delta$,
$$ \partdiff{F^{exp}}{x_{ij}} \Big|_{\bx'} \geq
\partdiff{F^{exp}}{x_{ij}} \Big|_{\bx} - \sum_{i,j} |x'_{ij} - x_{ij}|
\max \Big| \mixdiff{F^{exp}}{x_{ij}}{x_{i'j'}} \Big| 
\geq \partdiff{F^{exp}}{x_{ij}} \Big|_{\bx} - \delta m n M.$$
Consequently,
\begin{eqnarray*}
 F^{exp}(\bx + \delta (\by-\bx)) & \geq & F^{exp}(\bx) + \delta
 (\by-\bx) \cdot (\nabla F^{exp}(\bx) - \delta m n M \b1) \\ 
& \geq & F^{exp}(\bx) + \delta (\by-\bx) \cdot \nabla F^{exp}(\bx) -
 \delta^2 m^2 n^2 M \\ 
& \geq & F^{exp}(\bx) + \delta \cdot \frac14 \epsilon M - \delta^2 m^2
 n^2 M. 
\end{eqnarray*}
Again using $\delta = \frac{\epsilon}{8 m^2 n^2}$, we obtain 
$$ F^{exp}(\bx + \delta (\by-\bx)) \geq F^{exp}(\bx) +
 \frac{\epsilon^2}{32 m^2 n^2} M - \frac{\epsilon^2}{64 m^2 n^2} M 
 \geq F^{exp}(\bx) + \frac{\epsilon^2}{64 m^2 n^2}  M.$$
\end{proof}

\begin{lemma}
The number of iterations is with high probability at most $64 m^3 n^2
/ \epsilon^2$. 
\end{lemma}

\begin{proof}
By the previous lemma,  with high probability
the value of $F^{exp}(\bx)$ increases in each
iteration by at least $\frac{\epsilon^2}{64 m^2 n^2}  M$.
After $64 m^3 n^2 / \epsilon^2$ iterations, it will be at least $m
M$. By the definition of~$M$ and submodularity of valuations, $mM$ is
an upper bound on the welfare of every feasible allocation, and hence
also of the function~$F^{exp}$.  This completes the proof.
\end{proof}

This concludes the proof of Theorem~\ref{thm:approx-MIDR}. 
We remark that, building on the allocation rule in \cite{Dughmi11},
a similar approach gives a $(1-\epsilon)$-MIDR and
$(1-1/e-\epsilon)$-approximate allocation rule for combinatorial
public projects with weighted matroid rank sums that runs in
polynomial time in the value oracle model.

\section{From Approximately MIDR to Approximately TIE}
\label{se:MIDR->TIE}

In this section, we prove Theorem~\ref{thm:MIDR->TIE}. We assume that
we have a $(1-\epsilon)$-MIDR allocation rule $\cal M$ providing a
$c$-approximation for combinatorial auctions with valuations in a
class $\cal C$. We assume in the following that $c \geq \frac{1}{n}$,
where $n$ is the number of bidders. (A $\frac{1}{n}$-approximation is
trivial to achieve by giving all of the items to a random bidder.)
We also assume that $\epsilon = 1 / poly(m,n)$.
We want to convert the $(1-\epsilon)$-MIDR allocation rule into an
$(1-\epsilon')$-TIE mechanism. Our approach is as follows. 
If $\epsilon=0$, then the VCG payment scheme turns an MIDR mechanism
into a TIE mechanism. The fact that our mechanism is only
approximately MIDR means that the VCG payments might suffer from
errors that are significant for certain bidders, especially if
their utility is close to zero. Therefore, we modify the mechanism to
ensure that bidders whose valuation is very low do not participate in
the VCG scheme. In addition, we provide each bidder with the bundle of
{\em all} items with some small probability, so that their expected
utility is not extremely small. Our mechanism works as follows.

\paragraph{Mechanism ${\cal M}'$.}
\begin{enumerate}

\item Let $\cal M$ be a $(1-\epsilon)$-MIDR allocation rule.

\item Let $V_i$ be the valuation that bidder $i$ reports for the
  ground set of all items. 
Run $\cal M$ to compute a distribution over allocations
  $(S_1,\ldots,S_n)$ and let $O_i$ be an (unbiased) estimate of the
  expected value collected by bidder $i$, $\E[v_i(S_i)]$.\footnote{By
  polynomially bounded sampling, we can assume that our estimate $O_i$
  is with high probability within $\E[v_i(S_i)] \pm V_i / poly(m,n)$,
  and the probability of deviation decays exponentially. Similarly for
  the estimates of $O$ and $O'_{-i}$.} 
Let $O = \sum_{i=1}^{n} O_i$ be an estimate of the expected social
  welfare $\sum_{i=1}^{n} \E[v_i(S_i)]$. For each bidder $i$, run
  $\cal M$ also on the same instance without bidder $i$, and denote by
  $O'_{-i}$ an estimate of the expected social welfare of its
  outcome. 

\item Call bidder $i$ {\em relevant} if  $$V_i > \frac{1}{n^7} \sum_{j \neq
  i} V_j.$$ 

Call bidder $i$ {\em active} if he is relevant and, in addition,
$$\left(1-\frac{1}{n}\right) (O - O'_{-i}) + \frac{1}{2n^2} V_i >
\frac{1}{n^4} O'_{-i}.$$ 

\item With probability $1-1/n$: allocate a set $S_i$ from the
  distribution found by $\cal M$ to each active bidder $i$, and charge
  the VCG-like price $p_i = O'_{-i} - \sum_{j \neq i} O_j$. 
Do not allocate or charge anything to inactive bidders.

\item Else, with probability $1/n^2$ for each bidder $i$: If active,
  allocate the ground set to~$i$
and charge $\frac{1}{n^2} O'_{-i}$. If inactive, allocate the ground
  set with probability $1/2$ and charge $0$. 

\end{enumerate}

We emphasize that $O$ and $O'_{-i}$ are random variables
that we obtain by running the (randomized) allocation rule $\cal M$. 
We denote the actual optima over the range, with respect to the
reported valuations, by $OPT$ and $OPT'_{-i}$. 
In expectation, we have $OPT \geq \E[O] \geq (1-\epsilon) OPT$ and
$OPT'_{-i} \geq \E[O'_{-i}] \geq (1-\epsilon) OPT'_{-i}$; however with
some probability, $O$ could be significantly different from $OPT$
(even larger, since it is a probabilistic estimate), and 
$O'_{-i}$ could be significantly different from $OPT'_{-i}$.

In the following, we denote by $v^*_i$ the actual valuation of
bidder $i$, and by $V^*_i = v^*_i(M)$ the actual value of the ground set for bidder $i$.
Let $OPT_{+i}$ denote the optimum over the range with valuations $v_j$ for $j \neq i$ 
and $v^*_i$ for bidder $i$. (Note that $OPT_{+i} = OPT$ if $v^*_i = v_i$.)
Let $O_{+i}$ denote our estimate of $OPT_{+i}$ (assuming bidder $i$ reports the truth).
We prove the following statements. 

\begin{lemma}
\label{lem:small-bidder}
For every bidder such that $V^*_i \leq \frac{1}{n^7} \sum_{j \neq i}
V_j$, his expected utility is maximized 
within an $\epsilon$-fraction of his utility
by reporting truthfully.
\end{lemma}

\begin{proof}
Observe that as long as bidder $i$ reports $V_i \leq \frac{1}{n^7}
\sum_{j \neq i} V_j$, he is inactive and receives the same utility
regardless of his bid. Therefore, his utility could change only if he
reports 
$V_i > \frac{1}{n^7} \sum_{j \neq i} V_j$. In that case, he might be
classified as active (depending on 
$O$ and $O'_{-i}$). However, if that happens then he is charged at
least $\frac{1}{n^2} O'_{-i}$ with probability $\frac{1}{n^2}$,
i.e. $\frac{1}{n^4} O'_{-i}$ in expectation (conditioned on the value
of $O'_{i}$). 
Since the most value he can ever collect is $V^*_i$, he would (possibly)
gain from being active only if 
$O'_{-i} < n^4 V^*_i \leq \frac{1}{n^3} \sum_{j \neq i} V_j$. Since
$\E[O'_{-i}] \geq (1-\epsilon) OPT'_{-i} 
 > \frac{1}{n^2} \sum_{j \neq i} V_j$, it is very unlikely that
 $O'_{-i}$ is less than 
  $\frac{1}{n^3} \sum_{j \neq i} V_j$; 
this happens with exponentially small probability. Hence, by lying,
bidder $i$ could possibly gain only an exponentially small fraction of
$V^*_i$ in expectation, negligible with respect to his utility as an
inactive player. 
\end{proof}

\begin{lemma}
\label{lem:err-bound}
For every bidder $i$ such that $V^*_i > \frac{1}{n^7} \sum_{j \neq i} V_j$, we have 
$$ \E[|OPT_{+i} - O_{+i}|] \leq 3 \epsilon n^7 V^*_i $$
and
$$ \E[|OPT'_{-i} - O'_{-i}|] \leq 2 \epsilon n^7 V^*_i.$$
\end{lemma}

\begin{proof}
We have $V^*_i > \frac{1}{n^7} \sum_{j \neq i} V_j  \geq \frac{1}{n^7} OPT'_{-i}$.
We also know that $|OPT'_{-i} - \E[O'_{-i}]| \leq \epsilon OPT'_{-i} \leq \epsilon n^7 V^*_i$. 
The estimate $O'_{-i}$ of the output of the mechanism for all bidders except $i$
is concentrated around its expectation $\E[O'_{-i}]$, with variance $\frac{1}{poly(m,n)} \sum_{j \neq i} V_j
 \ll \epsilon n^7 V^*_i$, hence we can estimate
$$ \E[|OPT'_{-i} - O'_{-i}|] \leq 2 \epsilon n^7 V^*_i.$$

Similarly, $|OPT_{+i} - \E[O_{+i}]| \leq \epsilon OPT_{+i} \leq \epsilon (OPT_{-i} + V^*_i)
 \leq 2 \epsilon n^7 V^*_i$, and $O_{+i}$ is concentrated around its expectation,
therefore $\E[|OPT_{+i} - O_{+i}|] \leq 3 \epsilon n^7 V^*_i$.
\end{proof}

\begin{lemma}
\label{lem:big-bidder}
Every bidder~$i$ such that $V^*_i > \frac{1}{n^7} \sum_{j \neq i} V_j$
maximizes his expected utility within a factor of $(1-O(\epsilon n^9))$ by reporting his true valuation. 
\end{lemma}

\begin{proof}
Let us fix the valuations of all bidders except $i$.
Let us assume for now that our estimate $O_j$ is exactly equal to $\E[v_j(S_j)]$,
$O = \sum_{j=1}^{n} O_j$ is equal to $OPT$ (meaning that the MIDR mechanims
optimizes exactly), and similarly $O'_{-i}$ is equal to $OPT'_{-i}$.
We will analyze this idealized mechanism first. 

If bidder $i$ ends up being active, his expected utility will be
\begin{eqnarray*}
U_{active} & = & (1-1/n) (\E[v^*_i(S_i)] - (O'_{-i} - \sum_{j \neq i} O_j)) + \frac{1}{n^2} V^*_i - \frac{1}{n^4} OPT'_{-i} \\ 
& = & (1-1/n) (\E[v^*_i(S_i)] + \sum_{j \neq i} \E[v_j(S_j)] - OPT'_{i}) + \frac{1}{n^2} V^*_i - \frac{1}{n^4} OPT'_{-i} \\ 
& \leq & (1-1/n) (OPT_{+i} - OPT'_{-i}) + \frac{1}{n^2} V^*_i - \frac{1}{n^4} OPT'_{-i} = U^+_{active}.
\end{eqnarray*}
Here we used the fact that $OPT_{+i}$ is the optimal value over the range with valuations $v^*_i$ and $v_j, j \neq i$.
This implies that the last quantity, $U^+_{active}$, is the best possible utility bidder $i$ could receive as an active bidder.
In fact he will receive this utility if he reports truthfully and ends up being active.

If  bidder $i$ is inactive, then his expected utility will be
$$ U_{inactive} = \frac{1}{2n^2} V^*_i.$$

Now, if it is the case that 
$$ U^+_{active} - U_{inactive} =
(1-1/n) (OPT_{+i} - OPT'_{-i}) + \frac{1}{2n^2} V^*_i - \frac{1}{n^4} OPT'_{-i} \leq 0,$$ 
this means that no matter what bidder $i$ reports, being active cannot be more
profitable than not being active for him. When reporting his true valuation, such a bidder
will in fact be inactive,
because the condition for making a bidder active is exactly $(1-1/n) (O - O'_{-i}) + \frac{1}{2n^2} V_i
 - \frac{1}{n^4} O'_{-i} > 0$, and in this case we would have $OPT_{+i} = O$ and $O'_{-i} = OPT'_{-i}$.
Other than making the bidder inactive, the particular valuation he reports doesn't have an impact on his utility,
so he might as well report the truth.

On the other hand, if
$$ U^+_{active} - U_{inactive} = (1-1/n) (OPT_{+i} - OPT'_{-i}) + \frac{1}{2n^2} V^*_i - \frac{1}{n^4} OPT'_{-i} > 0,$$ 
then it is more profitable for bidder $i$ to be active,
since by reporting truthfully he will get utility $U^+_{active}$,
better than $U_{inactive}$ as an inactive bidder.
In fact we argued above that an active bidder cannot get a better utility by reporting
any valuation, so the best strategy for him is to report truthfully. In conclusion,
the idealized mechanism rewards a truthfully reporting bidder by utility $\max \{ U^+_{active}, U_{inactive} \}$
and that's the best the bidder can possibly receive.

Finally, we have to deal with the fact that 
$O = \sum_{j=1}^{n} O_j$ is not exactly equal to $OPT$,
and $O'_{-i}$ is not exactly equal to $OPT'_{-i}$.
By Lemma~\ref{lem:err-bound}, the estimates $O_{+i}$ and $O'_{-i}$ are in expectation within $O(\epsilon n^7 V^*_i)$ of $OPT_{+i}$ and $OPT_{-i}$. The estimates $\sum_{j \neq i} O_j$ of $\sum_{j \neq i} \E[v_j(S_j)]$ are strongly concentrated,
let's say with high probability within $\epsilon \sum_{j \neq i} V_j = O(\epsilon n^7 V^*_i)$ of the expectation.
Also, the actual social welfare of the distribution returned by the mechanism when bidder $i$ reports truthfully is within $O(\epsilon n^7 V^*_i)$ of $OPT_{+i}$ in expectation. 
Therefore, the expected utility of a truthfully reporting bidder is at least $\max \{ U^+_{active}, U_{inactive} \} - O(\epsilon n^7 V^*_i)$. On the other hand, the expected utility under any other reported valuation cannot be better than $\max \{U^+_{active}, U_{inactive}\} + O(\epsilon n^7 V^*_i)$, again due to the precision of the estimates stated above.
We have also ensured that the bidder's utility is at least $\frac{1}{2n^2} V^*_i$.
Therefore, the relative error in utility maximization is $O(\epsilon n^9)$. 
\end{proof}

Now we can prove Theorem~\ref{thm:MIDR->TIE}.

\begin{proof}
Using a $(1-\epsilon)$-MIDR mechanism $\cal M$, we implement a new
mechanism ${\cal M}'$ as above. 
By Lemma~\ref{lem:small-bidder} and \ref{lem:big-bidder}, each bidder
maximizes his utility within 
a factor of $1-O(\epsilon \cdot poly(n))$ by reporting
truthfully. Moreover, the expected social welfare 
provided by mechanism ${\cal M}'$ is at least $(1-1/n)$ times the
social welfare of all {\em active} 
bidders in $\cal M$. (Just considering the option that we used the VCG-based allocation.) 
It remains to estimate the loss in social welfare due to inactive bidders.

Consider a bidder $i$ such that $V_i \geq \frac{4}{n^2} OPT'_{-i}$. Since $OPT'_{-i} \geq \E[O'_{-i}]$,
and $O'_{-i}$ is a strongly concentrated estimate, with high probability we also have
 $V_i \geq \frac{2}{n^2} O'_{-i}$. Therefore, with high probability the bidder will be active and 
participate in the VCG scheme. The only bidders who do not participate with significant probability
are those such that $V_i < \frac{4}{n^2} OPT'_{-i} \leq \frac{4}{n^2} OPT$. Therefore,
all such bidders together cannot amount to more than $\frac{4}{n} OPT$. Overall, we recover
at least a $(1-O(1/n))$-fraction of the social welfare achieved by mechanism $\cal M$,
which means an approximation factor at least $(1-O(1/n)) c$. It is easy to see that
the $O(1/n)$ term can be replaced by any inverse polynomial in $m,n$, if desired.
\end{proof}

\section{\#P-hardness of Lottery Value Queries}
\label{sec:SharpP}

Here we show that for matroid rank functions, lottery-value queries are \#P-hard to answer, and require an exponential number of queries if the matroid is given by an independence oracle. We note that a lottery-value query for the vector $\bx = (\frac12, \ldots, \frac12)$
is simply an expectation over a uniformly random set of elements.

\begin{theorem}
There is a class of succintly represented matroids for which it is \#P-hard to compute $\E[r_\cM(R)]$, where $r_\cM$ is the rank function of $\cM$ and $R$ is a uniformly random set of matroid elements. For matroids given by an independence oracle, computing $\E[r_\cM(R)]$ requires an exponential number of queries.
\end{theorem}

\begin{proof}
We use the class of ``paving matroids''~\cite{Oxley}: Let $E$ be a
ground set of size $2m$ partitioned into disjoint pairs $e_1, e_2,
\ldots, e_m$, and let $\cF \subset {[m] \choose k}$ be any family of
$k$-element subsets of $[m]$. Then the following is a matroid: 
$S \subseteq E$ is independent iff either $|S| < 2k$, or $|S| = 2k$
and $S$ is {\em not} a union of pairs $\bigcup_{i \in F} e_i$ 
where $F \in \cF$.

Using this construction, we can embed any \#P-hard problem in a paving
matroid $\cM$. For example, consider the problem of counting perfect
matchings. For a graph $G$ with $m$ edges and $n$ vertices, we let $k
= n/2$ and we define $\cF$ to be the family of $k$-edge subsets of
edges that form a perfect matching. Then the matroid $\cM$ defined as
above captures the structure of perfect matchings, since for any set
of edges $F$ the rank function $r_{\cM}(F)$ tells us whether $F$ is a
perfect matching (which is the case if and only if $r_{\cM}(\bigcup_{i
  \in F} e_i) = 2|F|-1 = 2k-1$). Also, the matroid is succintly
represented by the graph $G$, in the sense that given $G$ we can
easily decide whether a given set is independent in $\cM$ or not. The
value of $r_{\cM}(S)$ depends on the structure of $S$ only if $S$ is a
union of $k$ pairs $e_i$, otherwise it is $r_{\cM}(S) = \min
\{|S|,2k\}$. Therefore, if we can compute the value of $\E[r_{\cM}(R)]
= 2^{-2m} \sum_{S \subseteq [2m]} r_{\cM}(S)$, we can extract the
number of perfect matchings by an elementary formula. 

Similarly, for a paving matroid given by an independence oracle, the
value of $\E[r_{\cM}(R)]$ determines the size of the family $\cF$,
which could be any arbitrary family. We cannot compute this value
unless we determine the size of $\cF$, which requires querying all
sets of the form $\bigcup_{i \in F} e_i$. This requires an exponential
number of queries for independence in $\cM$.
\end{proof}

\end{document}